\documentclass[11pt]{amsart} 
\usepackage{amsmath} 
\usepackage{amssymb}
\usepackage{harpoon} 
\usepackage[dvips]{epsfig} 
\theoremstyle{plain}
\newtheorem{theorem}{Theorem}[section] 
\newtheorem{lemma}[theorem]{Lemma}
\newtheorem{prop}[theorem]{Proposition}

\theoremstyle{definition}

\numberwithin{equation}{section}

\begin{document}

\title[Inequalities between size, mass, angular momentum, and charge] {Inequalities between size,
mass, angular momentum, and charge for axisymmetric bodies and the formation of trapped surfaces}

\author[Khuri]{Marcus Khuri} \address{Department of Mathematics\\ Stony Brook University\\ Stony
Brook, NY 11794, USA} \email{khuri@math.sunysb.edu}

\author[Xie]{Naqing Xie} \address{School of Mathematical Sciences\\ Fudan University\\ Shanghai
200433, PR China} \email{nqxie@fudan.edu.cn}

\dedicatory{Dedicated to the memory of our late friend and colleague Sergio Dain.}

\thanks{M. Khuri acknowledges the support of NSF Grant DMS-1308753. N. Xie is partially supported
by the National Science Foundation of China Grants No. 11421061, No. 11671089.}

\begin{abstract}
We establish inequalities relating the size of a material body to its mass,
angular momentum, and charge, within the context of axisymmetric initial data sets for the
Einstein equations. These inequalities hold in general without the assumption of the maximal
condition, and use a notion of size which is easily computable. Moreover, these results give rise to black hole existence criteria which are meaningful even in the time-symmetric case, and also include certain boundary effects.
\end{abstract} \maketitle

\section{Introduction}
\label{sec1} \setcounter{equation}{0} \setcounter{section}{1}

Inequalities relating the size of black holes\footnote{This refers to apparent horizons.} to the angular momentum and charge that they contain have been studied extensively and optimal results have been obtained, see \cite{ClementJaramilloReiris} and the references therein. More recently Dain \cite{Dain0.1,Dain} proposed extending these type of inequalities to arbitrary material bodies, and progress has been made in this direction \cite{AngladaDainOrtiz,Dain1,Khuri1,Khuri2,Reiris}. There are, however, undesirable features associated with each of these works. For instance, most are constrained to apply in the time-symmetric or maximal setting, whereas those which do not have this hypothesis involve complicated notions of size such as the Schoen/Yau radius \cite{SchoenYau2}. A primary motivation for the present article is to remove the time-symmetric and maximal assumptions, as well as to utilize simple notions of size such as lengths, areas, volumes etc.

These problems are naturally related to the Trapped Surface and Hoop Conjectures \cite{Seifert, Thorne}, which seek a mathematical formulation of the folklore belief that if enough matter and/or gravitational energy are present in a sufficiently small region, then the system must collapse to a black hole.  In \cite{Khuri1,Khuri2} this paradigm has been generalized to show that black hole formation may arise from concentration of angular momentum or charge, although the measurement of size in these results relies on the Schoen/Yau radius. In this paper we will present new results concerning the existence of trapped surfaces due to concentration of matter, angular momentum, and charge which rely on elementary measurements of size. Although the Trapped Surface and Hoop Conjectures have received much attention, they are far from being completely resolved. For instance, most works in this area impose strong hypotheses such as spherical symmetry, maximal slicing, or conformally flat geometry \cite{BeigOMurchadha,BizonMalecOMurchadha1,BizonMalecOMurchadha2,
Flanagan,Khuri0,Malec1,Malec2,MalecXie,OMurchadhaTungXieMalec,Wald1}. On the other hand, Schoen and Yau \cite{Eardley,SchoenYau2,Yau} were able to remove these restrictions, but paid for this generality with an intricate measurement of size and/or the fact that their results are not meaningful for low extrinsic curvature. Here we will mediate between these two sides by assuming the milder restriction of axisymmetry and obtaining results that are relevant even in the time-symmetric case.


\section{Setting and Statement of Results}
\label{sec2} \setcounter{equation}{0}
\setcounter{section}{2}

Let $(M, g, k)$ be an initial data set for the Einstein equations. This comprises a 3-dimensional manifold $M$, a Riemannian metric $g$, and a symmetric 2-tensor $k$ denoting the second fundamental form of an embedding into spacetime. If $\mu$ and $J$ are the energy and momentum densities of the matter fields, respectively, then the initial data satisfy the constraint equations
\begin{align}\label{0}
\begin{split}
16\pi \mu &= R+(Tr_{g} k)^{2}-|k|^{2},\\
8\pi J^{i} &= \nabla_{j}(k^{ij}-(Tr_{g} k)g^{ij}).
\end{split}
\end{align}
When charge plays a significant role, it is useful to extract the contributions of the electromagnetic field $(E,B)$ and consider the energy and momentum densities of the remaining fields
\begin{align}\label{1}
\begin{split}
\mu_{EM} &= \mu-\frac{1}{8\pi}\left(|E|^{2}+|B|^{2}\right),\\
J_{EM} &= J +\frac{1}{4\pi}E\times B,
\end{split}
\end{align}
where $(E\times B)_{i}=\epsilon_{ijl}E^{j}B^{l}$ is the cross product and $\epsilon$ is the volume form of $g$.

The initial data are said to be axially symmetric if the isometry group of $(M,g)$ contains a subgroup isomorphic to $U(1)$, and all quantities defining the initial data remain invariant under the $U(1)$ action. The associated axisymmetric Killing field will be denoted by $\eta$, and we have
\begin{equation}\label{2}
\mathfrak{L}_{\eta}g=\mathfrak{L}_{\eta}k=0,\ \mathfrak{L}_{\eta}E=\mathfrak{L}_{\eta}B=\mathfrak{L}_{\eta}J=0,\ \mathfrak{L}_{\eta}\mu=0,
\end{equation}
where $\mathfrak{L}_{\eta}$ is Lie differentiation. One advantage in the axisymmetric setting is the existence \cite{Chrusciel} of a global (cylindrical) Brill coordinate system $(\rho,z,\phi)$ in which the metric takes the simple form
\begin{equation}\label{3}
g=e^{-2U+2\alpha}(d\rho^{2}+dz^{2})+\rho^{2}e^{-2U}(d\phi+A_{\rho}d\rho+A_{z}dz)^{2}.
\end{equation}
Here the Killing vector is given by $\eta=\partial_{\phi}$, and thus all of the functions $U$, $\alpha$, $A_{\rho}$, and $A_{z}$ are independent of $\phi$. The proof of existence of Brill coordinates was carried out in \cite{Chrusciel} for simply connected, asymptotically flat initial data. In fact a byproduct of the proof asserts that such manifolds must be topologically trivial if they only have one end, that is $M\cong\mathbb{R}^3$. As the results of the current paper will be concerned solely with compact subsets of a simply connected $M$, we may obtain Brill coordinates by extending the initial data to the asymptotically flat regime. In particular we may assume without loss of generality that the quotient $\Sigma$ is diffeomorphic to the half plane.

In order to better understand the structure of the metric \eqref{3}, let us recall basic facts about the quotient manifold $\Sigma=M/U(1)$; this is the collection of all orbits of the $U(1)$ action and comes equipped with a canonical projection $\Pi:M\rightarrow M/U(1)$. The quotient metric is defined in the following way. Given $\hat{X},\hat{Y}\in T_{\hat{p}}\Sigma$, let $X,Y\in T_{p}M$ (with $\hat{p}=\Pi(p)$) be the unique vectors which are perpendicular to $\eta$
and satisfy $d\Pi_{p}(X)=\hat{X}$, $d\Pi_{p}(Y)=\hat{Y}$. The quotient metric is then defined by
\begin{equation}\label{4}
\hat{g}(\hat{X},\hat{Y})=g(X,Y)-\frac{g(X,\eta)g(Y,\eta)}{|\eta|^{2}}.
\end{equation}
In terms of the Brill coordinate expression,
the quotient metric is the first part of \eqref{3} in isothermal coordinates. The second part of \eqref{3} represents the square of the dual 1-form to $|\eta|^{-1}\eta$.
Note also that the Levi-Civita connection on $\Sigma$ is defined by
\begin{equation}\label{5}
\hat{\nabla}_{\hat{X}}\hat{Y}=d\Pi_{p}\left[(\nabla_{X}Y)^{\perp}\right],
\end{equation}
where $\perp$ is used to indicate the part perpendicular to $\eta$ and $\nabla$ is the Levi-Civita connection on $M$.


A body $\Omega$ is a connected open subset of $M$ with compact closure and smooth boundary $\partial\Omega$, and is referred to as axisymmetric if the $U(1)$ symmetry of $M$ acts by isometries on $\Omega$.  Axisymmetry allows for a suitable and well-defined notion of angular momentum for bodies (not necessarily axisymmetric), namely
\begin{equation}\label{6}
\mathcal{J}(\Omega)=\int_{\Omega}J_{i}\eta^{i} d\omega_{g}.
\end{equation}
In this setting gravitational waves do not carry angular momentum, so that all angular momentum arises from the matter fields. Without this assumption quasi-local angular momentum is a challenge to define \cite{Szabados}. Let us also record the square of the total charge
contained in the body
\begin{equation}\label{7}
Q^{2}=\left(\frac{1}{4\pi}\int_{\Omega}\operatorname{div} E d\omega_{g}\right)^{2}+\left(\frac{1}{4\pi}\int_{\Omega}\operatorname{div} B d\omega_{g}\right)^{2}.
\end{equation}

Recall that the gravitational field's strength near a
2-surface $S\subset M$ can be assessed by the quantities
\begin{equation}\label{8}
\theta_{\pm}:=H_{S}\pm Tr_{S}k,
\end{equation}
which are referred to as null expansions, where $H_{S}$ denotes the mean curvature in the
outward direction. The null expansions describe the rate at which the area of a shell of light is changing, when emitted outwardly by the surface in the future direction
($\theta_{+}$), and past ($\theta_{-}$).  Hence
the gravitational field is strong near
$S$ if $\theta_{+}< 0$ or $\theta_{-}< 0$, in which case
$S$ is called a future (past) trapped surface. We will say that a surface $S$ is (strongly) untrapped if the following (strict) inequality holds $H_{S}-|Tr_{S}k|\geq 0$.
Future (past) apparent horizons satisfy $\theta_{+}=0$
($\theta_{-}=0$), and arise from the boundaries of future
(past) trapped regions. Apparent horizons are interpreted as quasi-local versions of event horizons, and assuming Cosmic Censorship, they must generically be located inside
black holes \cite{Wald}.

In order to state the main results, we describe here certain energy conditions and lengths that will be used. Let $(e_{1},e_{2},e_{3}=|\eta|^{-1}\eta)$ be an orthonormal frame field on $M$ and $\vec{J}=J(e_{1})e_{1}+J(e_{2})e_{2}$, then
\begin{equation}\label{9}
\mu\geq|\vec{J}|+|J(e_{3})|
\end{equation}
is a slightly stronger version of the standard dominant energy condition
\begin{equation}\label{10}
\mu\geq|J|=\sqrt{|\vec{J}|^{2}+J(e_{3})^{2}}.
\end{equation}
When charge is present, the following the so called charged dominant energy condition will be useful
\begin{equation}\label{11}
\mu_{EM}\geq |J_{EM}|.
\end{equation}
Furthermore, we set $l=2\pi\min_{\Omega}|\eta|$ and $L=2\pi\max_{\Omega}|\eta|$ to be the minimum and maximum lengths of axisymmetric orbits in $\Omega$, and will always assume that $l>0$ so that $\Omega$ stays away from the axis. Lastly $\chi(\Pi(\Omega))$ will denote the Euler characteristic of the projection of $\Omega$ in the orbit space $\Sigma$.

Our first result yields a lower bound for size in terms of the energy-momentum content of the matter fields, and shows that if too much energy-momentum is contained in a domain of a fixed size then collapse must ensue. Notice also that a contribution is played by the null expansions of the boundary, similar to the inequalities of Yau in \cite{Yau}.

\begin{theorem}\label{thm1}
Let $(M,g,k)$ be an axially symmetric, simply connected initial data set.
Let $\Omega\subset M$ be an axisymmetric body situated away from the axis $(l>0)$, satisfying the dominant energy condition $\mu\geq|J|$, and with strongly untrapped boundary.\smallskip

$i)$ If $\Omega$ is void of apparent horizons then
\begin{equation}\label{12}
\int_{\Omega}(\mu-|J|)d\omega_{g}+\frac{L}{8\pi l}\int_{\partial\Omega}
(H-|Tr_{\partial\Omega}k|)d\sigma_{g}\leq\frac{L^2}{4l}\chi(\Pi(\Omega)).
\end{equation}

$ii)$ If the opposite (strict) inequality in \eqref{12} holds then $\Omega$ must contain an apparent horizon.
\end{theorem}

It is interesting to note that the mild hypotheses of this theorem imply a strong topological restriction on $\Omega$. Namely from \eqref{12} and the strictly untrapped nature of the boundary, the Euler characteristic of the projection must be positive $\chi(\Pi(\Omega))>0$. Hence by the classification theorem for surfaces \cite{Munkres}, we have that $\Pi(\Omega)$ is homeomorphic to a disk or rather $\Omega$ is homeomorphic to a solid torus.

The next two results replace the role of total energy-momentum in Theorem \ref{thm1} by total angular momentum and total charge. They show a bound for the amount of angular momentum and charge that can be possessed by a body, in terms of its size. Alternatively, if quantum effects are taken into account, these inequalities reveal a minimum possible size for a spinning and/or charged particle \cite{Khuri1,Khuri2} (see also \cite{ArnowittDeserMisner}). Furthermore, a black hole existence criterion is given based on concentration of these two quantities.

\begin{theorem}\label{thm2}
Let $(M,g,k)$ be an axially symmetric, simply connected initial data set.
Let $\Omega\subset M$ be an axisymmetric body situated away from the axis $(l>0)$, satisfying the enhanced dominant energy condition \eqref{9}, and with strongly untrapped boundary.\smallskip

$i)$ If $\Omega$ is void of apparent horizons then
\begin{equation}\label{13.1}
|\mathcal{J}(\Omega)|+\frac{L}{16\pi^2}\int_{\partial\Omega}
(H-|Tr_{\partial\Omega}k|)d\sigma_{g}\leq\frac{L^2}{8\pi}\chi(\Pi(\Omega)).
\end{equation}

$ii)$ If the opposite (strict) inequality in \eqref{13.1} holds then $\Omega$ must contain an apparent horizon.
\end{theorem}

\begin{theorem}\label{thm3}
Let $(M,g,k)$ be an axially symmetric, simply connected initial data set.
Let $\Omega\subset M$ be an axisymmetric body situated away from the axis $(l>0)$, satisfying the strict dominant energy condition $\mu>|J|$ in $\Omega$ and the charged dominant energy condition $\mu_{EM}\geq|J_{EM}|$ on $\partial\Omega$, and with strongly untrapped boundary. \smallskip

$i)$ If $\Omega$ is void of apparent horizons then
\begin{equation}\label{14.1}
Q^2+\frac{\mathcal{C}_{0}L|\partial\Omega|^2}{16\pi^{2} l|\Omega|}\int_{\partial\Omega}
(H-|Tr_{\partial\Omega}k|)d\sigma_{g}
\leq\frac{\mathcal{C}_{0}L^2|\partial\Omega|^2}{8\pi l|\Omega|}\chi(\Pi(\Omega)),
\end{equation}
where
\begin{equation}
\mathcal{C}_{0}=\frac{\max_{\partial\Omega}(\mu-|J_{EM}|)}
{\min_{\Omega}(\mu-|J|)}.
\end{equation}

$ii)$ If the opposite (strict) inequality in \eqref{14.1} holds then $\Omega$ must contain an apparent horizon.
\end{theorem}

Although the inequalities \eqref{12}, \eqref{13.1}, and \eqref{14.1} remain valid in the time-symmetric and maximal cases, the content of the corresponding apparent horizon exist statements becomes vacuous in this setting. This is due primarily to the method of proof, which utilizes a version of the Jang equation \cite{SchoenYau1}. In particular, the apparent horizons are produced by forcing solutions of the Jang equation, with constant Dirichlet boundary conditions, to blow-up. However, in the time-symmetric and maximal cases these solutions are always smooth. Therefore part $(ii)$ of the above theorems is only relevant when sufficient amounts of extrinsic curvature are involved. This situation is analogous to the range of applicability of the black hole existence results in \cite{Eardley,SchoenYau2,Yau} which also rely on the Jang equation. For this reason we provide refined versions of Theorems \ref{thm1}, \ref{thm2}, and \ref{thm3} in Section \ref{sec4} for the time-symmetric and maximal cases, which do not rely on this method of proof and give meaningful statements in this setting.

\section{Stability of the Quotient}
\label{sec3} \setcounter{equation}{0}
\setcounter{section}{3}

The proofs of the main results are based on a type of stability property of the quotient manifold $\Sigma$ along with properties of the generalized Jang equation \cite{BrayKhuri1,BrayKhuri2}. In this section we discuss the stability property, which was first observed by Reiris in \cite{Reiris}. Our contribution here is to refine Reiris' observation, give an independent proof, and to elucidate the meaning from a purely initial data perspective (as opposed to a spacetime point of view). The goal of this section is to establish the following result.

\begin{prop}\label{prop1}
Let $(M,g)$ be an axisymmetric, simply connected, 3-dimensional Riemannian manifold, and let $\Omega\subset M$ be a compact axisymmetric body. If $R$ denotes scalar curvature and $H$ denotes the mean curvature of $\partial\Omega$ with respect to the unit outward normal then
\begin{equation}\label{15}
\int_{\Pi(\Omega)}\frac{1}{2}R dA\leq 2\pi\chi(\Pi(\Omega))
-\int_{\partial\Pi(\Omega)}H ds.
\end{equation}
where $dA$ and $ds$ are elements of area and arc length with respect to the quotient metric $\hat{g}$.
\end{prop}

Inequality \eqref{15} is referred to as a stability type inequality in analogy with the stability inequality for minimal hypersurfaces, noting that the right-hand side arises from the Gaussian curvature. Let $K$ denote the Gaussian curvature of the quotient $\Sigma$, then according to the Gauss-Bonnet formula
\begin{equation}\label{16}
\int_{\Pi(\Omega)}KdA=2\pi\chi(\Pi(\Omega))-\int_{\partial\Pi(\Omega)}\kappa ds,
\end{equation}
where $\kappa$ is geodesic curvature. Recall that if
$\hat{\gamma}\in\Sigma$ is the curve (parametrized by arc length) representing
$\partial\Pi(\Omega)$, then its geodesic curvature is given by
\begin{equation}\label{17}
\kappa=\hat{g}(\hat{\gamma}',\hat{\nabla}_{\hat{\gamma}'}\hat{\nu}),
\end{equation}
where $\hat{\nu}$ is the unit outer normal to $\partial\Pi(\Omega)$.
Let $\gamma\in\partial\Omega$
(parametrized by arc length) be the curve such that $\Pi(\gamma)=\hat{\gamma}$ and $\gamma'\perp\eta$, then
$d\Pi(\gamma')=\hat{\gamma}'$ so the notation is consistent; also let $\nu$ be the unit normal to $\partial\Omega$ with $d\Pi(\nu)=\hat{\nu}$. It follows that $(\gamma',|\eta|^{-1}\eta,\nu)$ forms an orthonormal frame on $\partial\Omega$, and we may evaluate the mean curvature by
\begin{equation}\label{18}
H=g(\gamma',\nabla_{\gamma'}\nu)
+g\left(\frac{\eta}{|\eta|},\nabla_{\frac{\eta}{|\eta|}}\nu\right).
\end{equation}
Next observe that since $\eta$ is a Killing field
\begin{equation}\label{19}
g(\eta,\nabla_{\eta}\nu) =-g(\nabla_{\eta}\eta,\nu)
=g(\nabla_{\nu}\eta,\eta) =\frac{1}{2}\nu(|\eta|^{2}),
\end{equation}
and thus
\begin{equation}\label{20}
H=g(\gamma',\nabla_{\gamma'}\nu)
+\nu(\log|\eta|)
\end{equation}
which implies that
\begin{equation}\label{21}
\kappa
=g(\gamma',\nabla_{\gamma'}\nu)
=H-\nu(\log|\eta|)=H-\hat{\nu}(\log|\eta|).
\end{equation}
Therefore \eqref{16} becomes
\begin{equation}\label{22}
\int_{\Pi(\Omega)}KdA
=2\pi\chi(\Pi(\Omega))-\int_{\partial\Pi(\Omega)}Hds
+\int_{\Pi(\Omega)}\hat{\Delta}\log|\eta|dA.
\end{equation}
It remains to compute the last integral in \eqref{22}. For this we give an independent proof of Proposition 2.1 in \cite{Reiris}, from the initial data perspective.

\begin{lemma}\label{lemma1}
Under the hypotheses of Proposition \ref{prop1}
\begin{equation}\label{23}
\hat{\Delta}|\eta|=|\eta|\left(K-\frac{1}{2}R\right)
-\frac{1}{4}|\eta|^{3}e^{4(U-\alpha)}(\partial_{z}A_{\rho} -\partial_{\rho}A_{z})^2,
\end{equation}
where $\alpha$, $U$, $A_{\rho}$, and $A_{z}$ are coefficients in the Brill form \eqref{3} of the metric $g$.
\end{lemma}

\begin{proof}
Recall that for two 2-dimensional Riemannian metrics conformally related by $g_{2}=e^{2v}g_{1}$, the relation between the corresponding Gaussian curvatures and Laplacians is
\begin{equation}\label{24}
\Delta_{g_{1}}v=K_{g_{1}}-e^{2v}K_{g_{2}}, \quad\quad\quad \Delta_{g_{2}}=e^{-2v}\Delta_{g_{1}}.
\end{equation}
Applying this to the present situation where $\hat{g}=e^{-2U+2\alpha}\delta$ produces \begin{equation}\label{25}
\Delta_{\rho,z}(\alpha-U)=-e^{-2U+2\alpha}K,
\end{equation}
where $\delta=d\rho^2+dz^2$ and $\Delta_{\rho,z}=\partial_{\rho}^2+\partial_{z}^2$.
Next observe that the scalar curvature in Brill coordinates \cite[(3.7)]{Chrusciel} takes a particularly simple form
\begin{equation}\label{26}
2e^{-2U+2\alpha}R=8\Delta U-4\Delta_{\rho,z}\alpha -4|\nabla
U|_{\delta}^2-\rho^2 e^{-2\alpha}(\partial_{z}A_{\rho}-\partial_{\rho}A_{z})^2,
\end{equation}
where $\Delta$ is the flat Laplacian on $\mathbb{R}^3$ given by
\begin{equation}
\Delta=\Delta_{\rho,z}+\frac{1}{\rho}\partial_{\rho}.
\end{equation}
Therefore subtracting 4
times \eqref{25} from \eqref{26}, so that $\Delta_{\rho,z}\alpha$ is eliminated, yields
\begin{equation}\label{27}
\Delta_{\rho,z}U=\frac{1}{2}e^{-2U+2\alpha}R-\frac{2}{\rho}\partial_{\rho}U +|\nabla
U|_{\delta}^2+\frac{\rho^2}{4}e^{-2\alpha}(\partial_{z}A_{\rho}-\partial_{\rho}A_{z})^2 -e^{-2U+2\alpha}K.
\end{equation}
Lastly, applying the second relation in \eqref{24} along with direct computation shows that
\begin{align}\label{28}
\begin{split}
\hat{\Delta}|\eta|=& e^{2U-2\alpha}\Delta_{\rho,z}(\rho e^{-U})\\
=&e^{2U-2\alpha}\left(2\partial_{\rho}e^{-U}
+\rho\partial_{\rho}^{2}e^{-U}+\rho\partial_{z}^{2}e^{-U}\right)\\
=&\rho
e^{U-2\alpha}\left(-\Delta_{\rho,z}U+|\nabla U|_{\delta}^2 -\frac{2}{\rho}\partial_{\rho}U\right).
\end{split}
\end{align}
Inserting \eqref{27} into \eqref{28} yields the desired result.
\end{proof}

In order to complete the proof of Proposition \ref{prop1}, use Lemma \ref{lemma1} in \eqref{22} to obtain
\begin{align}\label{29}
\begin{split}
\int_{\Pi(\Omega)}\frac{1}{2}RdA
=&2\pi\chi(\Pi(\Omega))-\int_{\partial\Pi(\Omega)}Hds\\
&-\int_{\Pi(\Omega)}\left(|\hat{\nabla}\log|\eta||^2+\frac{1}{4}|\eta|^2e^{4(U-\alpha)}
(\partial_{z}A_{\rho}-\partial_{\rho}A_{z})^2\right)dA.
\end{split}
\end{align}

\section{Proof of the Main Results}
\label{sec4} \setcounter{equation}{0}
\setcounter{section}{4}

In order to establish Theorems \ref{thm1}, \ref{thm2}, and \ref{thm3} we will utilize properties of the generalized Jang equation \cite{BrayKhuri1,BrayKhuri2}, which are described below. Motivated by the desire to impart nonnegativity to the scalar curvature of initial data satisfying the dominant energy condition, or to study the existence of apparent horizons, one may solve the following equation (for $f$)
\begin{equation}\label{30}
\left(g^{ij}-\frac{u^2 f^{i}f^{j}}{1+u^2|\nabla f|^{2}}\right)\left(\frac{u\nabla_{ij}f+u_{i}f_{j}+u_{j}f_{i}}
{\sqrt{1+u^2|\nabla f|^{2}}}-k_{ij}\right)=0,
\end{equation}
where $f^{i}=g^{ij}\nabla_{j}f$ and $u$ is a given positive function. Geometrically,
solutions should be viewed as a graph $t=f(x)$ in the warped product 4-manifold $(\mathbb{R}\times M, u^2 dt^2+g)$, where equation \eqref{30} is equivalent to the apparent horizon equation if $k$ is extended to the 4-manifold appropriately. When regular solutions do not exist, the graph $t=f(x)$ blows-up and approximates a cylinder over an apparent horizon in the base manifold $M$ (see \cite{BrayKhuri1,BrayKhuri2}). On the other hand, if a bounded domain $\Omega\subset M$ does not contain any apparent horizons and the boundary is strongly untrapped $H-|Tr_{\partial\Omega}k|>0$, then the Dirichlet problem for \eqref{30} with $f=0$ on $\partial\Omega$ admits a smooth solution \cite{HanKhuri,Yau}; the solution will be axisymmetric as long as $u$ is axisymmetric, that is $\eta(u)=0$ implies $\eta(f)=0$.
Furthermore, the scalar curvature of the induced graph metric $\overline{g}=g+u^2 df^2$ has the expression \cite{BrayKhuri1,BrayKhuri2}
\begin{equation}\label{31}
\overline{R}=16\pi(\mu-J(w))+
|h-k|_{\overline{g}}^{2}+2|q|_{\overline{g}}^{2}
-2u^{-1}\operatorname{div}_{\overline{g}}(uq),
\end{equation}
where $h$ is the second fundamental form of the graph in the dual Lorentzian setting $(\mathbb{R}\times M, -u^{2}dt^2+\overline{g})$,
and $w$ and $q$ are 1-forms given by
\begin{equation}\label{32}
h_{ij}=\frac{ u \nabla_{ij}f
+ u_i f_j +  u_j  f_i}{ \sqrt{1 + u^2 |\nabla f|^2 }},\text{ }\text{ }\text{ }\text{ }
w_{i}=\frac{u f_{i}}{\sqrt{1+u^{2}|\nabla f|^{2}}},\text{
}\text{ }\text{ }\text{ }
q_{i}=\frac{u f^{j}}{\sqrt{1+u^{2}|\nabla f|^{2}}}(h_{ij}-k_{ij}).
\end{equation}

Choose $u=|\eta|^{-1}$ and suppose that the generalized Jang equation has a regular solution over $\Omega$, which is the case when no apparent horizons are present. The first task is to estimate the matter-energy content of the body. According to \eqref{31}, $\mu-J(w)\geq\mu-|J|\geq 0$, and $\overline{g}\geq g$ we have that
\begin{align}\label{33}
\begin{split}
\int_{\Omega}(\mu-J(w))|\eta|d\omega_{g}
&\leq\int_{\Omega}(\mu-J(w))|\eta|d\omega_{\overline{g}}\\
&\leq\left(\frac{L}{2\pi}\right)^{2}\frac{1}{8\pi}
\int_{\Omega}|\eta|^{-1}\left(\frac{1}{2}\overline{R}
+u^{-1}div_{\overline{g}}(uq)\right)d\omega_{\overline{g}}\\
&=\left(\frac{L}{2\pi}\right)^{2}\frac{1}{8\pi}\left(
\int_{\Omega}\frac{1}{2}|\eta|^{-1}\overline{R}d\omega_{\overline{g}}
+\int_{\partial\Omega}|\eta|^{-1}q(\overline{\nu})d\sigma_{\overline{g}}\right),
\end{split}
\end{align}
where $\overline{\nu}$ is the unit outer normal to $\partial\Omega$ with respect to $\overline{g}$. Now note that $\eta$ is a Killing field for $\overline{g}$ with $|\eta|_{g}=|\eta|_{\overline{g}}$, and we may extend $\overline{g}$ to be axisymmetric on all of $M$. It then follows from the discussion
in Section \ref{sec2} that there is a global system of Brill coordinates $(\overline{\rho},\overline{z},\phi)$ for $\overline{g}$. In particular
\begin{equation}\label{34}
d\omega_{\overline{g}}=|\eta|d\phi\wedge d\overline{A},\text{
}\text{ }\text{ }\text{ }\text{ } \text{ }\text{ }\text{ }d\sigma_{\overline{g}}=|\eta|d\phi\wedge
d\overline{s},
\end{equation}
and Proposition \ref{prop1} applies to yield
\begin{equation}\label{35}
\int_{\Omega}(\mu-J(w))|\eta|d\omega_{g}
\leq\frac{1}{4}\left(\frac{L}{2\pi}\right)^{2}
\left(2\pi\chi(\Pi(\Omega)) -\int_{\partial\Pi(\Omega)}(\overline{H}-q(\overline{\nu}))d\overline{s}\right).
\end{equation}
Next, note that a calculation in \cite[Page 764]{Yau} shows that
\begin{equation}\label{36}
\overline{H}-q(\overline{\nu})=\sqrt{1+u^{2}|\nabla
f|^{2}}(H+\xi Tr_{\partial\Omega}k)
-\xi \frac{Tr_{\partial\Omega}k}{u|\nabla f|+\sqrt{1+u^2|\nabla f|^2}},
\end{equation}
where $\xi(x)=\pm 1$ exactly when $\nu(f)(x)=\mp|\nabla f(x)|$; recall that $f|_{\partial\Omega}=0$ so that this defines $\xi$ at each point of the boundary. Since by assumption
$H>|Tr_{\partial\Omega}k|\geq 0$ it follows that
\begin{equation}\label{37}
\overline{H}-q(\overline{\nu})\geq H-|Tr_{\partial\Omega}k|,
\end{equation}
and hence with the help of $d\overline{s}\geq ds$ we find that
\begin{equation}\label{38}
\int_{\Omega}(\mu-J(w))|\eta|d\omega_{g}
\leq\frac{L^2}{8\pi}\chi(\Pi(\Omega)) -\frac{L}{16\pi^2}
\int_{\partial\Omega}(H-|Tr_{\partial\Omega}k|)d\sigma_{g}.
\end{equation}

We are now in a position to prove the results stated in Section \ref{sec2}. First note that
Theorem \ref{thm1} $(i)$ is a direct consequence of \eqref{38} together with
\begin{equation}\label{39}
\int_{\Omega}(\mu-J(w))|\eta|d\omega_{g}\geq\frac{l}{2\pi}\int_{\Omega}(\mu-|J|)d\omega_{g}.
\end{equation}
In order to establish Theorem \ref{thm2} $(i)$ observe that since the Jang solution $f$ is axisymmetic $w(\eta)=0$, and thus $|J(w)|\leq|\vec{J}|$. Hence, with the enhanced dominant energy condition $\mu\geq|\vec{J}|+|J(e_{3})|$ we have
\begin{align}\label{40}
\begin{split}
|\mathcal{J}(\Omega)|\leq\int_{\Omega}|J(\eta)|d\omega_{g}
&=\int_{\Omega}|J(e_{3})||\eta|d\omega_{g}\\
&=\int_{\Omega}\left[|J(e_{3})|+|\vec{J}|-\mu
+(\mu-|\vec{J}|)\right]|\eta|d\omega_{g}\\
&\leq\int_{\Omega}(\mu-|\vec{J}|)|\eta|d\omega_{g}\\
&\leq\int_{\Omega}(\mu-J(w))|\eta|d\omega_{g}.
\end{split}
\end{align}
By combining \eqref{38} and \eqref{40} the desired result follows.
Next observe that if the charged dominant energy condition $\mu_{EM}\geq|J_{EM}|$ holds on the boundary and $\mu>|J|$ is valid on the interior of $\Omega$ then
\begin{align}\label{41}
\begin{split}
Q^{2}
=&\left(\frac{1}{4\pi}\int_{\partial\Omega}E\cdot\nu d\sigma_{g}\right)^{2}
+\left(\frac{1}{4\pi}\int_{\partial\Omega}B\cdot\nu d\sigma_{g}\right)^{2}\\
\leq&\frac{|\partial\Omega|}{16\pi^{2}}
\int_{\partial\Omega}\left(|E|^{2}+|B|^{2}\right)d\sigma_{g}\\
=&\frac{|\partial\Omega|}{16\pi^{2}}\int_{\partial\Omega}\left[\left(|E|^{2}+|B|^{2}
-8\pi\mu+8\pi|J_{EM}|\right)+8\pi(\mu-|J_{EM}|)\right]d\sigma_{g}\\
\leq&\frac{|\partial\Omega|}{2\pi}\int_{\partial\Omega}(\mu-|J_{EM}|)d\sigma_{g}\\
\leq&\frac{|\partial\Omega|^2\mathcal{C}_{0}}{2\pi|\Omega|}
\int_{\Omega}(\mu-|J|)d\sigma_{g}.
\end{split}
\end{align}
This together with Theorem \ref{thm1} $(i)$ produces Theorem \ref{thm3} $(i)$.

Lastly, part $(ii)$ of all three theorems arises from the following argument. Assume by way of contradiction that $\Omega$ is void of apparent horizons. Then a regular solution to the generalized Jang equation exists as described at the beginning of this section. Therefore we may apply the above arguments to establish that part $(i)$ is valid, yielding a contradiction to the hypotheses of part $(ii)$.

\section{The Time-Symmetric and Maximal Cases}
\label{sec5} \setcounter{equation}{0}
\setcounter{section}{5}

As discussed at the end of Section \ref{sec2}, reliance on Jang type equations to establish Theorems \ref{thm1}, \ref{thm2}, and \ref{thm3} ensures that in the maximal ($Tr_{g}k=0$) and time-symmetric cases ($k=0$) part $(ii)$ of these theorems, concerning black hole existence, is not meaningful. Therefore, in this setting, we seek alternate criteria which guarantee the presence of a trapped surface. Consider the inequality of Proposition \ref{prop1}, rewritten from the perspective of $\Omega$:
\begin{equation}\label{42}
\int_{\Omega}R|\eta|^{-1}d\omega_{g}\leq 8\pi^2\chi(\Pi(\Omega))
-2\int_{\partial\Omega}H|\eta|^{-1}d\sigma_{g}.
\end{equation}
In the time-symmetric and maximal cases, the dominant energy condition implies nonnegative scalar curvature $R\geq 0$. This together with nonnegative mean curvature $H\geq 0$ yields a refined version of Theorem \ref{thm1} $(i)$. In addition, \eqref{42} leads to conditions which force $\partial\Omega$ to be trapped or partially trapped. In this regard we define a surface $S$ to be averaged trapped if
\begin{equation}\label{43}
\int_{S}H|\eta|^{-1}d\sigma_{g}< 0.
\end{equation}
Note that this definition in axisymmetry differs slightly from previous definitions of averaged trapped surfaces (see eg. \cite{MalecXie}) by the presence of $|\eta|^{-1}$ in the surface measure. We also set the oscillation of a function to be $\operatorname{osc}(f)=\max f-\min f$.

\begin{theorem}\label{thm4}
Let $(M,g)$ be an axially symmetric, simply connected initial data set.
Let $\Omega\subset M$ be an axisymmetric body situated away from the axis $(l>0)$, with nonnegative scalar curvature $R\geq 0$.\smallskip

$i)$ If $\partial\Omega$ has nonnegative mean curvature $H\geq 0$ then
\begin{equation}\label{44}
\int_{\Omega}Rd\omega_{g}+2\int_{\partial\Omega}Hd\sigma_{g}\leq 4\pi L\chi(\Pi(\Omega)).
\end{equation}

$ii)$ If
\begin{equation}\label{45}
\int_{\Omega}Rd\omega_{g}> 4\pi L\chi(\Pi(\Omega)),
\end{equation}
then $\partial\Omega$ is an averaged trapped surface.\medskip

$iii)$ If
\begin{equation}\label{46}
\int_{\Omega}Rd\omega_{g}> 4\pi L\chi(\Pi(\Omega))
+\frac{L}{\pi}|\partial\Omega|\operatorname{osc}(H|\eta|^{-1}),
\end{equation}
then $\partial\Omega$ is a trapped surface.
\end{theorem}

\begin{proof}
Observe that $(ii)$ follow immediately from \eqref{42}, and $(i)$ follows similarly with the aid of $|\eta|^{-1}\geq 2\pi L^{-1}$. In order to establish $(iii)$, assume by way of contradiction that $\partial\Omega$ is not trapped. Then $\max(H|\eta|^{-1})\geq 0$, and hence \eqref{42} implies
\begin{equation}\label{47}
\int_{\Omega}Rd\omega_{g}\leq 4\pi L\chi(\Pi(\Omega))
+\frac{L}{\pi}\int_{\partial\Omega}\left[\max(H|\eta|^{-1})-H|\eta|^{-1}\right]d\sigma_{g},
\end{equation}
which yields the opposite inequality in \eqref{46}, a contradiction.
\end{proof}

Recall that in time-symmetry the quantity $\int_{\Omega}Rd\omega_{g}$ is referred to as the rest mass of $\Omega$. Therefore the criteria in $(ii)$ and $(iii)$ of the above theorem are consistent with the basic intuition that highly concentrated mass in a region of fixed size results in gravitational collapse. Analogous results may be obtained for concentration of angular momentum and charge in the maximal setting. To see this observe that slight modification of the inequalities \eqref{40}, under the dominant energy condition assumption, produces
\begin{equation}\label{48}
|\mathcal{J}(\Omega)|\leq\frac{L}{32\pi^2}\int_{\Omega}Rd\omega_{g}.
\end{equation}
Moreover, under the assumption of nonnegative non-electromagentic matter energy $\mu_{EM}\geq 0$ on $\partial\Omega$, similar arguments to those of \eqref{41} yield
\begin{equation}\label{49}
Q^2\leq\frac{|\partial\Omega|^2\mathcal{C}_{1}}{32\pi^2|\Omega|}
\int_{\Omega}Rd\omega_{g},
\end{equation}
where
\begin{equation}\label{49.1}
\mathcal{C}_{1}=\frac{\max_{\partial\Omega}R}{\min_{\Omega}R}.
\end{equation}
We then have the following analogues of Theorems \ref{thm2} and \ref{thm3}.

\begin{theorem}\label{thm5}
Let $(M,g,k)$ be an axially symmetric, simply connected, maximal initial data set.
Let $\Omega\subset M$ be an axisymmetric body situated away from the axis $(l>0)$, satisfying the dominant energy condition $\mu\geq|J|$.\smallskip

$i)$ If $\partial\Omega$ has nonnegative mean curvature $H\geq 0$ then
\begin{equation}\label{50}
|\mathcal{J}(\Omega)|+\frac{L}{16\pi^2}\int_{\partial\Omega}Hd\sigma_{g}
\leq \frac{L^2}{8\pi}\chi(\Pi(\Omega)).
\end{equation}

$ii)$ If
\begin{equation}\label{51}
|\mathcal{J}(\Omega)|> \frac{L^2}{8\pi}\chi(\Pi(\Omega)),
\end{equation}
then $\partial\Omega$ is an averaged trapped surface.\medskip

$iii)$ If
\begin{equation}\label{52}
|\mathcal{J}(\Omega)|> \frac{L^2}{8\pi}\chi(\Pi(\Omega))
+\frac{L^2}{32\pi^3}|\partial\Omega|\operatorname{osc}(H|\eta|^{-1}),
\end{equation}
then $\partial\Omega$ is a trapped surface.
\end{theorem}

\begin{theorem}\label{thm6}
Let $(M,g,k)$ be an axially symmetric, simply connected, maximal initial data set.
Let $\Omega\subset M$ be an axisymmetric body situated away from the axis $(l>0)$, satisfying positive scalar curvature $R>0$ in $\Omega$ and nonnegative non-electromagnetic matter energy $\mu_{EM}\geq 0$ on $\partial\Omega$.\smallskip

$i)$ If $\partial\Omega$ has nonnegative mean curvature $H\geq 0$ then
\begin{equation}\label{53}
Q^2+\frac{\mathcal{C}_{1}|\partial\Omega|^2}{16\pi^2|\Omega|}\int_{\partial\Omega}Hd\sigma_{g}
\leq \frac{\mathcal{C}_{1}L|\partial\Omega|^2}{8\pi|\Omega|}\chi(\Pi(\Omega)).
\end{equation}

$ii)$ If
\begin{equation}\label{54}
Q^2>\frac{\mathcal{C}_{1}L|\partial\Omega|^2}{8\pi|\Omega|}\chi(\Pi(\Omega)),
\end{equation}
then $\partial\Omega$ is an averaged trapped surface.\medskip

$iii)$ If
\begin{equation}\label{55}
Q^2> \frac{\mathcal{C}_{1}L|\partial\Omega|^2}{8\pi|\Omega|}\chi(\Pi(\Omega))
+\frac{\mathcal{C}_{1}L|\partial\Omega|^3}{32\pi^3|\Omega|}\operatorname{osc}(H|\eta|^{-1}),
\end{equation}
then $\partial\Omega$ is a trapped surface.
\end{theorem}

\begin{proof}
Part $(i)$ of Theorems \ref{thm5} and \ref{thm6} follows directly from \eqref{48} and \eqref{49} combined with Theorem \ref{thm4} $(i)$. Parts $(ii)$ and $(iii)$ are then established in a similar way to that of Theorem \ref{thm4}.
\end{proof}

Finally, we speculate on the heuristic physical reasoning behind the formation of black holes due to concentration of angular momentum and charge. Intuitively, large amounts of angular momentum should induce an object to fly apart rather than collapse, and large amounts of charge should cause a body's constituent parts to repel, again suggesting that such a body
will come apart rather than implode. Thus it seems at first thought to be counter-intuitive to expect a body to collapse due to concentration of angular momentum or charge.
However, it should be taken into account that in order for a body to stay together with
such high degrees of angular momentum or charge, it should be extremely massive to create a sufficiently powerful gravitational field to counteract the forces (angular momentum/charge) pulling it apart. With this we obtain a reasonable heuristic explanation for the collapse. Namely, sufficiently large amounts of angular momentum or charge in a body of a fixed size implies correspondingly large amounts of matter/energy in the same region, which according to the Trapped Surface and Hoop Conjectures should induce gravitational collapse. Moreover, this explanation is also consistent with the mathematical proofs above.

\end{document}